\documentclass[11pt,reqno]{amsart}
\usepackage{graphicx}
\usepackage{amsaddr}
\usepackage{setspace}
\usepackage{amsmath}
\usepackage{amssymb}
\usepackage{textcomp}

\usepackage{amsthm}

\usepackage{wrapfig}
\usepackage{booktabs}
\usepackage{bm}
\usepackage[english]{babel}
\usepackage[utf8]{inputenc}
\usepackage[T1]{fontenc}
\usepackage{xcolor}
\usepackage{epstopdf}
\usepackage{mathtools}
\usepackage{hyperref}
\usepackage{geometry}
\usepackage{verbatim}
\usepackage{relsize}
\usepackage{subfigure}
\usepackage{multicol}
\usepackage{nicefrac}
\usepackage{epigraph}
\usepackage{dsfont}
\usepackage{enumerate}
\usepackage{qcircuit}
\usepackage{cancel}
\usepackage{diagbox}
\usepackage{caption}

\usepackage{natbib}
\setcitestyle{aysep={}}
\usepackage[babel,english=british]{csquotes}

\newcommand{\IR}{\mathbb{R}}
\newcommand{\IN}{\mathbb{N}}

\newcommand{\IC}{\mathbb{C}}

\newcommand{\Pseudo}{\textrm{Pdim}}
\newcommand{\fat}{\textrm{fat}}

\newtheorem{theorem}{Theorem}[section]
\newtheorem*{theorem*}{Theorem}

\newtheorem{corollary}[theorem]{Corollary}

\newtheorem{lemma}[theorem]{Lemma}
\newtheorem{definition}[theorem]{Definition}

\theoremstyle{definition}

\newtheorem{remark}[theorem]{Remark}

\numberwithin{equation}{section}

\DeclareMathOperator{\tr}{tr}

\renewenvironment{proof}{{\bfseries Proof:}}{}

\title{Pseudo-dimension of quantum circuits}
\date{\today}
\author[Caro, Datta]{Matthias C. Caro$^{\dagger 1}$ \and Ishaun Datta$^{\ddagger 1,2}$}
\email{$\dagger$\href{mailto:caro@ma.tum.de}{caro@ma.tum.de}, $\ddagger$\href{mailto:idatta@stanford.edu}{idatta@stanford.edu}}
\address{$^1$Technical University of Munich, Germany, Department of Mathematics\\
$^2$Stanford University, USA, Institute for Computational and Mathematical Engineering}

\begin{document}

\maketitle

\begin{abstract}
We characterize the expressive power of quantum circuits with the pseudo-dimension, a measure of complexity for probabilistic concept classes. We prove pseudo-dimension bounds on the output probability distributions of quantum circuits; the upper bounds are polynomial in circuit depth and number of gates. Using these bounds, we exhibit a class of circuit output states out of which at least one has exponential gate complexity of state preparation, and moreover demonstrate that quantum circuits of known polynomial size and depth are PAC-learnable.
\end{abstract}
\makeatletter
    \providecommand\@dotsep{5}
\makeatother

\section{Introduction}
\label{intro}
An important line of research in classical learning theory is characterizing the expressive power of function classes using complexity measures. Such complexity bounds can in turn be used to bound the size of training data required for learning. Among the most prominent of these are the Vapnik-Chervonenkis (VC) dimension introduced by \citep{Vapnik.1971}. Other well-known measures are the pseudo-dimension due to \citep{Pollard.1984}, the fat-shattering dimension due to \citep{Alon.1997}, the Rademacher complexities \citep[see][]{Bartlett.2002}, and more generally covering numbers in metric spaces.
\par
The goal of characterizing an object’s expressive power also appears in different guises throughout quantum information. A well-known example is quantum state tomography. \citep{Aaronson.2007} related a variant of state tomography to a classical learning task whose fat-shattering dimension can be bounded using a particular function class related to the set of quantum states. Associated to this is a corresponding upper bound on sample complexity.\par
\citep{Aaronson.2007} observes that there is no analogous theorem for general quantum process tomography, but leaves as an open question whether there are restricted classes of operations that are information-efficiently learnable. We answer this question in the affirmative. In particular, we show that for quantum circuits with depth and size polynomial in the number of qubits, quantum process tomography is possible using only polynomially many examples. \par
Gate complexity of unitary implementation and state preparation are yet another example of how one may capture the richness of a function class that corresponds to a quantum computational process \citep[see e.g.][]{Aaronson.2016}. For unitary complexity, the challenge is to determine, e.g., how many two-qubit unitaries (i.e.~two-qubit logical gates, in a computational setting) are required to implement a certain multi-qubit unitary (i.e.~a quantum circuit). For the gate complexity of state preparation, it is to determine how many unitaries produce a certain multi-qubit state. An alternative perspective, adopted in this work, is to consider the expressive power of a set of circuits with a fixed number of unitaries.\par 
In this work we describe a new way of applying complexity measures from classical learning, specifically pseudo-dimension, to quantum information. We associate with a quantum circuit a natural probabilistic function class describing the outcome probabilities of measurements performed on the circuit output. In this way, a function class corresponding to a quantum circuit can be studied with the classical tool of pseudo-dimension. Here, we show that the pseudo-dimension of such a class can be bounded in terms of a polynomial of the circuit depth and size. We also give two applications of these bounds, one for the gate complexity of quantum state preparation, the other in learnability of quantum circuits.\par
These findings are noteworthy not only because of the results themselves, but because we demonstrate the power of pseudo-dimension to gain insight into quantum computation. We hope that these tools may be applied to other problems in quantum computing in future work.
\subsection{Related Work}
\citep{Aaronson.2007} showed that using the framework of PAC learning, one can introduce a variant of quantum state tomography and prove an upper bound on the required number of copies of the unknown state. This idea was developed further in \citep{Aaronson.2018b} and \citep{Aaronson.2018}.\par
Motivated by Aaronson's work, \citep{Cheng.2016} use pseudo-dimension and fat shattering dimension to characterize the learnability of measurements, as a dual problem to learning the state. We apply this mathematical framework to study the problem of learning the circuit itself, in particular by offering a natural function class corresponding to a quantum circuit.\par
\citep{Rocchetto.2017} proved that stabilizer states, prevalent in error correction, are \emph{computationally}-efficiently learnable, establishing a connection between efficient classical simulability and computationally efficient learnability. This was realized experimentally for small optical systems in \citep{Rocchetto.2019}. Similarly, in Sec.~\ref{sec:discussion} we pose as an open problem whether there are quantum operations that can be PAC-learned with modest computation, which could then in principle be demonstrated in an experiment.   \par
In \citep{Chung.20191103}, the authors study the problem of PAC-learning classes of functions with computational basis states as input and quantum output, possibly mixed. We highlight two main differences: first, whereas we assume the training data to be measurement statistics, \citep{Chung.20191103} consider examples given as classical-quantum states. Thus, the two scenarios are not directly comparable. Our learning result yields a semi-classical strategy for the problem described in \citep{Chung.20191103}, though it is possibly suboptimal. Second, the learnability result of \citep{Chung.20191103} is only for finite concept classes, whereas our result does not have this restriction. While \citep{Chung.20191103} show learnability of quantum circuits with a finite gate set, we allow for arbitrary $2$-qudit gates, i.e.~a continuous gate set. Note that our corresponding notions of learnability differ.\par 
While we take a formal approach to learning quantum circuits, others have studied learning unitaries numerically, e.g.~with heuristics such as gradient descent \citep{Kiani.20200213}. Practical machine learning algorithms have also been used for state tomography by \citep{Torlai.2018}, and similar techniques could be applied to restricted classes of process tomography. \par
Another branch of quantum learning deals with whether quantum examples can decrease the information-theoretic complexity of learning a classical function. There are different flavors of this question, e.g.~depending on whether learning is distribution-specific or distribution-independent. \citep{Arunachalam.2017} gives an overview of some of these aspects of quantum learning. 
\par
In classical learning theory, bounding the complexity measures of function classes (based on complexity-theoretic assumptions) has been studied widely. \citep{Goldberg.1995} derived an upper bound on the VC-dimension of a function class in terms of the runtime required by an algorithm implementing the elements of that class. \citep{Karpinski.1997} established an analogous bound for the function class implemented by a neural network (for various activation functions) in terms of the number of nodes and the number of programmable parameters of the network. \citep{Koiran.1996} demonstrated that by bounding the complexity of function classes implemented on a given architecture, one can lower-bound the size of an architecture implementing a specific ``hard'' function. 
\subsection{Overview of Results}
We consider the general scenario in which one measures the output state of a $2$-local qu$d$it quantum circuit, generating a probability distribution. 
We do not assume \emph{geometric} locality, i.e.~we do not assume that $2$-qu$d$it unitaries act on neighboring qu$d$its.
We show an upper bound on the pseudo-dimension of the distributions arising from these quantum circuits. By doing so, we provide insight into the complexity or ``hardness'' of the circuit and the output state that gives rise to the probability distribution. Below, we provide informal statements of the key results.
\begin{theorem*}[Pseudo-dimension bounds, Informal] \label{informal_pdbounds}\ \\\\
Consider quantum circuits with fixed architecture, namely those for which the input qudits of the $2$-qudit gates are specified, but the gates may vary subject to this constraint. That is, we allow for arbitrary $2$-qu$d$it unitaries, and in particular we do not restrict ourselves to a finite gate library.
\\Parameterize a quantum circuit $\mathcal{N}$ by its qudit dimension $d,$ depth $\delta,$ and number of gates or size $\gamma.$\\ 
\emph{\textbf{Theorem \ref{ThmPseudoDimFixedQuantumCircuit}}:} For a suitable function class $\mathcal{F}_\mathcal{N}$ corresponding to the possible probability distributions formed by product measurements in the computational basis on the circuit output, $\Pseudo(\mathcal{F}_\mathcal{N})\leq \mathcal{O}(d^4\cdot \gamma\log\gamma).$\\\\
Consider quantum circuits with variable architecture, i.e.~those for which the input qudits of the gates are not specified. For such circuits of depth $\delta$ and number of gates or size $\gamma,$ one may similarly define function classes $\mathcal{F}_{\delta,\gamma}$ for circuits whose gates are unitaries, and $\mathcal{G}_{\delta, \gamma}$ for circuits whose gates are quantum operations, which describe the possible probability distributions formed by product measurements on the circuit output. Then, \\
\emph{\textbf{Theorem \ref{ThmPseudoDimVariableQuantumCircuit}}:} $\Pseudo(\mathcal{F}_{\delta,\gamma})\leq \mathcal{O} (\delta \cdot d^4 \cdot \gamma^2 \log \gamma).$\\
\emph{\textbf{Theorem \ref{ThmPseudoDimVariableQuantumCircuitOperations}}:}
$\Pseudo(\mathcal{G}_{\delta,\gamma})\leq \mathcal{O} (\delta \cdot d^8 \cdot \gamma^2 \log \gamma).$
\end{theorem*}
\noindent All upper bounds are polynomial in the dimension $d$, the depth $\delta$, and the size $\gamma$.\par
In Sec.~\ref{sec:lower}, we demonstrate how to apply these complexity upper bounds to explicitly construct, for each $n\in\IN$, a finite-but-large set of $n$-qubit quantum states, out of which at least one cannot be implemented by a $2$-local qudit circuit of subexponential depth or size. \par
\begin{theorem*}[Gate Complexity of State Preparation, Informal]\ \\
\noindent For any subset $C\subseteq\lbrace |x0\rangle\rbrace_{x\in\lbrace 0,1\rbrace^n}$, define
\[
  |\psi_C\rangle =\begin{cases}
               \frac{1}{\sqrt{|C|}}\sum\limits_{|x0\rangle\in C} |x0\rangle\quad&\textrm{if } C\neq\emptyset\\
               |0\rangle^{\otimes n}\otimes |1\rangle &\textrm{if } C=\emptyset.
            \end{cases}
\]
\noindent If each state in $\{|\psi_C\rangle\}_{C}$ can be generated from the input state $|0\rangle^{\otimes (n+1)}$ by some circuit of depth $\delta$ and size $\gamma$, then $2^n \leq \mathcal{O}\left(\delta\cdot \gamma^2 \log \gamma \right).$
As a corollary, there exists at least one such $C$ so that $|\psi_C\rangle$ requires a circuit exponential in depth and size.
\end{theorem*}
Analogously to  \citep{Aaronson.2007}, in Sec.~\ref{learnability} we use our pseudo-dimension bounds to prove a relaxed variant of quantum process tomography, which following Aaronson’s terminology can be called \textit{pretty-good circuit tomography}: 

\begin{theorem*}[Learnability, Informal]
Given a circuit with depth $\Delta$ and size $\Gamma$, both polynomial in the number of qu$d$its and known in advance to the learner, polynomially-many training examples, each a triple of input state, output measurement, and corresponding probability, suffice to learn the quantum operation implemented by a $2$-local quantum circuit of depth $\Delta$ and size $\Gamma$.\\
That is, for confidence $\delta,$ accuracy, $\varepsilon,$ and error margins $\alpha$ and $\beta,$ all in $(0,1)$, a candidate circuit of depth $\Delta$ and size $\Gamma$ that performs sufficiently well (in a sense made rigorous in Sec.~\ref{learnability}) on  $$\mathcal{O}\left(\frac{1}{\varepsilon}\left( \Delta d^8\Gamma^2\log\Gamma\log^2\left(\frac{\Delta d^8\Gamma^2\log\Gamma}{(\beta - \alpha)\varepsilon}\right) + \log\frac{1}{\delta}\right) \right)$$ many samples will with probability at least $1-\delta$ approximate the actual circuit from which the samples are drawn.
\end{theorem*}
\noindent In this framework, each training example is a three-tuple of the input state, the observed measurement outcome, and the corresponding measurement probability. Alternately, one may take each training example as a two-tuple of the input state and the measurement outcome, whose probability is the corresponding measurement probability \citep[see][Appendix 8]{Aaronson.2007}.\\ \par

%
We review the basics of quantum information, quantum computation, and classical learning theory in Sec.~\ref{sec:preliminaries}. We also discuss prior classical results as motivation. Sec.~\ref{sec:results} contains our main results on the pseudo-dimension of quantum circuits and the respective proofs. In Sec.~\ref{sec:applications}, we apply these results to fin lower bounds on the gate complexity of quantum state preparation and to a learning problem for quantum operations. We conclude with open questions in Sec.~\ref{sec:discussion}.
\section{Preliminaries}\label{sec:preliminaries}

As our readership includes both physicists and computer scientists, in this section we review the mathematical frameworks of quantum information theory and learning theory. Further details appear in the reference texts \citep{Heinosaari.2013} and \citep{Nielsen.2010}.\par
\subsection{Quantum Information and Computation}
The most general descriptor of a $d$-level quantum system or statistical ensemble thereof is a density matrix, an element of $$\mathcal{S}\left( \IC^d\right):= \{ \rho\in\IC^{d\times d} ~|~ \rho\geq 0,\ \tr[\rho]=1\}.$$ Here, $\rho\geq 0$ means that the matrix $\rho$ is Hermitian and all its eigenvalues are non-negative. An important subset of density matrices is the set of pure states, which are one-dimensional projections. Following Dirac notation, we denote the projector onto the subspace spanned by a unit vector $|\psi\rangle\in\IC^d$ by $|\psi\rangle\langle\psi |$. By the spectral theorem, every quantum state can be written as a convex combination of pure states, though this decomposition is not unique in general.\par
Central to the framework of quantum mechanics is the measurement, the mechanism by which one may observe properties of a quantum system. These are typically described by so-called positive-operator valued measures (POVMs). As we focus on measurements with a finite set of outcomes $\{i\}$, it suffices to think of measurements as collections of so-called effect operators $\{ E_i\}_{i=1}^m$ with $E_i\in\IC^{d\times d}$, $0\leq E_i\leq \mathds{1}_d$, and $\sum\limits_{i=1}^m E_i = \mathds{1}_d$. We denote the set of effect operators by $$\mathcal{E}\left( \IC^d\right):=\{ E\in\IC^{d\times d} ~|~ 0\leq E_i\leq \mathds{1}_d\}.$$ Again, we highlight a special case: if we take an orthonormal basis $\{ |\psi_i\rangle\}_{i=1}^d$ of $\IC^d$, then the set $\{ E_i= |\psi_i\rangle\langle \psi_i|\}_{i=1}^d$ is called a projective measurement.\par
Born's rule connects measurements to measurement outcomes: given a state characterized by a density operator, the effect operator has a corresponding probability $p_i = \text{tr}[\rho E_i]$. Thus the requirement that the effect operators sum to the identity can be seen as probabilities summing to one. In the special case of pure state $\rho =|\psi\rangle\langle\psi|$ and projective measurement $\{ E_i= |\psi_i\rangle\langle \psi_i|\}_{i=1}^d$, the probability of outcome $i$ is $p_i = \tr[\rho E_i]=|\langle\psi|\psi_i\rangle|^2$.\par
So far we have described the components of static quantum theory. The dynamics of quantum states are described by so-called quantum operations, which we denote by 
\begin{align*} 
\mathcal{T}\left(\IC^d \right):= \{ & T:\IC^{d\times d}\to\IC^{d\times d} ~|~ T~ \text{is linear, } \text{completely positive, and trace-non-increasing}\}.
\end{align*} 
Here, a map $T$ is completely positive if $T\otimes Id_n$ is positivity-preserving for every $n\in\IN$. If $T\in\mathcal{T}\left(\IC^d \right)$ is trace-preserving, we call $T$ a quantum channel. An important example is the unitary quantum channel, $T(\rho) = U\rho U^*$ for some unitary $U\in\IC^{d\times d}$.\\
Note that any element of $\mathcal{T}\left(\IC^d \right)$ is a linear map between vector spaces of dimension $d^2$ and can thus be understood as a $d^2\times d^2$ matrix.\par

%
%
%
\subsection{Classical Learning Theory and Complexity Measures}
Next we describe the ``probably approximately correct'' (PAC) model of learning, introduced and formalized by \citep{Vapnik.1971} and \citep{Valiant.1984}. In (realizable) PAC learning for spaces $X$, $Y$ and a concept class $\mathcal{F}\subseteq Y^X$, a learning algorithm receives as input labeled training data $\lbrace (x_i,f(x_i))\rbrace_{i=1}^m$  for some $f\in\mathcal{F}$, where the samples $x_i$ are drawn independently according to some unknown probability distribution $D$ on $X$ that is unknown to the learner. Given the training examples, the goal of the learner is to approximate the unknown function $f$ by a hypothesis function $h,$ with high probability.\par
We can formalize this as follows: first, we introduce a loss function $\ell:Y\times Y\to\mathbb{R}_+$ to quantify the discrepancy between the hypothesis $h$ and the function $f$. 
We call a concept class $\mathcal{F}$ PAC-learnable if there exists a learning algorithm $\mathcal{A}$ such that for every probability distribution $D$ on $X$, $f\in\mathcal{F}$ and $\delta,\varepsilon\in (0,1)$, running $\mathcal{A}$ on training data drawn according to $D$ and $f$ yields a hypothesis $h$ such that $\mathbb{E}_{x\sim D}[\ell(h(x),f(x))]\leq\varepsilon$ with probability $\geq 1-\delta$ (with regard to the choice of training data). 
Moreover, we quantify the minimum amount of training data that an algorithm $\mathcal{A}$ needs to meet the above conditions by a map $m_\mathcal{F}:(0,1)\times (0,1)\to\IN$, $(\delta,\varepsilon)\mapsto m(\delta,\varepsilon) $, the so-called sample complexity of $\mathcal{F}$. We focus on proper learning, in which the learning algorithm must output as its hypothesis an element of the concept class, i.e., we require $h\in\mathcal{F}$.\par
A standard approach to assessing learnability is to characterize the complexity of the respective concept class $\mathcal{F}$. Many such complexity measures are used, the most common being the VC-dimension for binary-valued function classes $\mathcal{F}\subseteq \{0,1\}^X$, named after its progenitors \citep{Vapnik.1971}. This combinatorial parameter can be shown to fully characterize the learnability: a concept class $\mathcal{F}\subseteq \{0,1\}^X$ is PAC-learnable (w.r.t.~the $0$-$1$-loss) if and only if the VC-dimension of $\mathcal{F}$ is finite. Moreover, the sample complexity of PAC learning $\mathcal{F}$ can be expressed in terms of its VC-dimension \citep[see][]{Blumer.1989,Hanneke.2016}.\par
In this work, we employ a widely-used extension of the VC-dimension to real-valued concept classes:
\begin{definition}\emph{(Pseudo-dimension \citep{Pollard.1984})}\label{DffPseudoDim}
Let $\mathcal{F}\subseteq \IR^X$ be a real-valued concept class. A set $\lbrace x_1,...,x_k\rbrace\subseteq X$ is pseudo-shattered by $\mathcal{F}$ if there are $y_1,...,y_k\in\IR$ such that for any $C\subseteq\lbrace 1,...,k\rbrace$ there is an $f_C \in\mathcal{F}$ such that for all $1\leq i\leq k, $
$i\in C$ if and only if $f_C(x_i)\geq y_i.$\\

\noindent The pseudo-dimension of $\mathcal{F}$ is defined to be
\begin{align*}
\Pseudo (\mathcal{F}) := \sup\lbrace n\in\IN_0 ~|~\ &\exists S\subseteq X\text{ s.t. } |S|=n ~\text{and }\text{S is pseudo-shattered by } \mathcal{F}\rbrace.
\end{align*}
\end{definition}
\noindent Alternatively, one can express the pseudo-dimension in terms of the VC-dimension. Namely, 
\begin{align*}
    \Pseudo (\mathcal{F}) = \textrm{VC} (\left\lbrace X\times\IR\ni (x,y)\mapsto\textrm{sgn}(f(x)-y) ~|~ f\in\mathcal{F} \right\rbrace).
\end{align*}
Here, the VC-dimension for a function class $\mathcal{H}\subseteq\{\pm 1\}^Z$ is defined as 
\begin{align*}
    \textrm{VC} (\mathcal{H}) := \sup\lbrace n\in\IN_0 ~|~&\exists z_1,\ldots,z_n\in Z\text{ s.t. }\forall b\in\{\pm 1\}^n~\exists h_b\in\mathcal{H}\text{ s.t. }\forall i:h_b(z_i)=b_i\rbrace.
\end{align*}

\noindent There is also a scale-sensitive version of the pseudo-dimension:
\begin{definition}\emph{(Fat-Shattering Dimension \citep{Alon.1997})}
Let $\mathcal{F}$ be a real-valued concept class and let $\alpha >0$. A set $\lbrace x_1,...,x_k\rbrace
\subseteq X$ is $\alpha$-fat-shattered by $\mathcal{F}$ if there are $y_1,...,y_k\in\IR$ such that for any $C\subseteq\lbrace 1,...,k\rbrace$ there is an $f_C\in\mathcal{F}$ such that for all $1\leq i\leq k$:
\begin{enumerate}
\item $i\notin C$ $\Rightarrow$ $f_C(x_i)\leq y_i-\alpha$ and
\item $i\in C$ $\Rightarrow$ $f_C(x_i)\geq y_i +\alpha$.
\end{enumerate}
The $\alpha$-fat-shattering dimension of $\mathcal{F}$ is defined to be
\begin{align*}
\emph{\fat}_\mathcal{F}(\alpha) := \sup\lbrace n\in\IN_0|\ &\exists S\subseteq X\text{ s.t. } |S|=n \wedge \textrm{S} \text{ is } \alpha\textrm{-fat-shattered by }\mathcal{F}\rbrace.
\end{align*}
\end{definition}

\noindent Note that, trivially, $\emph{\fat}_\mathcal{F}(\alpha)\leq \Pseudo (\mathcal{F})$ holds for every $\alpha >0$ and for every real-valued function class $\mathcal{F}$.\\

Sample complexity upper bounds for $[0,1]$-valued function classes in terms of the fat-shattering dimension have been proved in \citep{Bartlett.1998, Anthony.2000}.
\section{Pseudo-Dimension Bounds for Quantum Circuits}\label{sec:results}
We now formulate how to characterize the expressive power of quantum circuits. 
In particular, we consider circuits with $n$ input registers of qu$d$its, size (i.e.~number of gates) $\gamma,$ and depth (i.e.~number of layers) $\delta$. More precisely, we consider circuits composed of two-qudit unitaries, i.e.~logical gates with two inputs. Note that two-qudit gates include one-qudit gates. We assume that gates in the same layer and acting on disjoint pairs of qudits can act in parallel. Additionally, we assume that each qu$d$it is acted upon by at least one gate, else it effectively does not participate in the circuit.\par
In this section, we assign function classes to quantum circuits and then derive bounds on the pseudo-dimension of these function classes, in terms of the number of qudits and the size and depth of the circuits. First, we fix quantum circuit structure and inputs, varying only the entries of the unitary gates and thereby the resulting function. Then, we broaden our scope to variable circuit architectures, variable inputs, and circuits whose `gates' are general quantum operations.\par
An important tool that will recur throughout our work is the following result on polynomial sign assignments, used in \citep{Goldberg.1995} to derive VC-dimension bounds from computational complexity.
\begin{theorem}\emph{\citep[][Theorem $3$]{Warren.1968}}\label{ThmWarrenNonZero}
Let $\lbrace p_1,\ldots,p_m\rbrace$ be a set of real polynomials in $n$ variables with $m\geq n$, each of degree at most $d\geq 1$. Then the number of consistent non-zero sign assignments to $\lbrace p_1,\ldots,p_m\rbrace$ is at most $\left( \frac{4edm}{n}\right)^n$.
\end{theorem}

\noindent Here, $e$ is Euler's number and a ``consistent non-zero sign assignment'' to a set of polynomials $\lbrace p_1,\ldots,p_m\rbrace$ is a vector $b\in\{ \pm 1\}^m$ s.t.~there exist $x_1,\ldots,x_n\in\IR$ for which it holds that $\textrm{sgn}(p_i(x_1,\ldots,x_n))=b_i$ for all $1\leq i\leq m$.\par
The following implication of Theorem \ref{ThmWarrenNonZero} for consistent but not necessarily non-zero sign assignments (which we define as above, but with $b\in\{-1,0,1\}^m$) to sets of polynomials was observed in \citep[][Corollary 2.1]{Goldberg.1995}.
\begin{corollary}\label{CrlWarren}
Let $\lbrace p_1,\ldots,p_m\rbrace$ be a set of real polynomials in $n$ variables with $m\geq n$, each of degree at most $d\geq 1$. Then the number of consistent sign assignments to $\lbrace p_1,\ldots,p_m\rbrace$ is at most $\left( \frac{8edm}{n}\right)^n$.
\end{corollary}
\noindent \begin{proof} (Sketch)
This can be obtained by applying Theorem \ref{ThmWarrenNonZero} to the set $\lbrace p_1+\varepsilon,p_1-\varepsilon,\ldots,p_m+\varepsilon,p_m-\varepsilon\rbrace$ with $\varepsilon>0$ chosen sufficiently small.
\qed \end{proof}
\subsection{Fixed Circuit Structure}\label{pseudo-fixed}
Suppose we fix the architecture of a quantum circuit of depth $\delta$ and size $\gamma$. Specifically, we restrict our attention to 2-local quantum circuits, i.e.~circuits whose logical gates have support on two qudits, not necessarily neighboring each other. (See Fig.~\ref{fig:circuit}.) ``Fixed architecture'' means that we specify the positions of the two-qudit unitaries, namely their order and which qudits they act on. Though the unitaries' positions are fixed, we may vary the entries of the unitaries themselves. Here, we allow for arbitrary $2$-qu$d$it unitaries. In particular, we do not restrict ourselves to a finite gate library. Can we bound the pseudo-dimension of the function class of measurement probability distributions that this circuit generates? And how does the bound depend on $d$ (the dimensionality of the qudits), $\delta$ and $\gamma$?\par
To formalize this question: let $n\in\IN$ be the number of qudits, $d\in\IN$ be their dimensionality, and $\mathcal{N}$ be a fixed quantum circuit architecture of depth $\delta$ and size $\gamma$ acting on $n$ qudits. We enumerate the positions of the two-qudit unitaries in $\mathcal{N}$ by tuples $(i,j)$ with $1\leq i\leq \delta$ denoting the layer and $1\leq j\leq \gamma_i$ the position of the unitary among all the unitaries inside layer $i$, where w.l.o.g.~we count from top to bottom and take into account only the first qudit on which a unitary acts.\par 
Note that  $\sum\limits_{i=1}^\delta \gamma_i=\gamma$, and trivially $\gamma_i\leq\gamma$ and $\gamma_i\leq\frac{n}{2}$, as we assume that every qudit is acted upon by at least one gate. We write the unitary at position $(i,j)$ as $U^{(i,j)}$. These constitute the ``free parameters'' which we can vary in order to make the quantum circuit perform different tasks. The overall unitary implemented by $\mathcal{N}$ when plugging in the unitaries $\lbrace U^{(i,j)}\rbrace_{1\leq i\leq \delta, 1\leq j\leq \gamma_i}$ at the respective positions we denote by $U_{\mathcal{N}|\lbrace U^{(i,j)}\rbrace}$. Note that $U_{\mathcal{N}|\lbrace U^{(i,j)}\rbrace}$ strongly depends on the two-qudit unitaries that are plugged into the architecture, but sometimes we will suppress this dependence and simply write $U_\mathcal{N}$ for notational ease.\par
The quantum circuit $\mathcal{N}$ now gives rise to the following set of output states:
\begin{align*}
\mathcal{S}_\mathcal{N}\left( (\IC^d)^{\otimes n}\right):=\lbrace U_{\mathcal{N}|\lbrace U^{(i,j)}\rbrace}|0\rangle^{\otimes n} |\ U^{(i,j)}\in\mathcal{U}\left( (\IC^d)^{\otimes 2}\right)\rbrace.
\end{align*}
These output states in turn give rise to a function class of measurement probability distributions with regard to product measurements:
\begin{align*}
\mathcal{F}_\mathcal{N}:= \lbrace f:X\to [0,1]\ |\ \exists |\psi\rangle\in \mathcal{S}_\mathcal{N}\left( (\IC^d)^{\otimes n}\right):~f(x)=|\langle x|\psi\rangle|^2\rbrace,
\end{align*}
where we take $X=S_d\times\ldots\times S_d$ to be the Cartesian product of $n$ unit spheres of $\IC^d$.\par
%
%
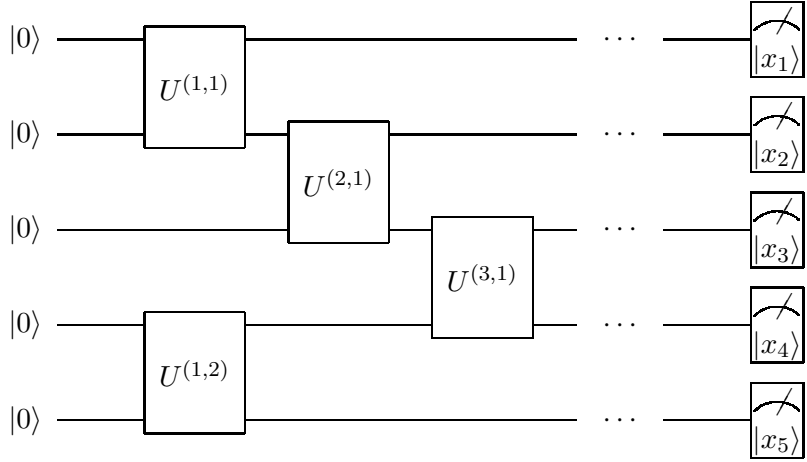
\begin{figure}[h]
\centering \leavevmode
\mbox{
\Qcircuit @C=1.5em @R=0.7em {
\lstick{|0\rangle} & \qw & \multigate{1}{U^{(1,1)}} & \qw & \qw & \qw & \cdots & & \qw & \meterB{|x_1\rangle}\\
\lstick{|0\rangle} & \qw & \ghost{U^{(1,1)}} & \multigate{1}{U^{(2,1)}} & \qw & \qw & \cdots & & \qw & \meterB{|x_2\rangle}\\
\lstick{|0\rangle} & \qw & \qw & \ghost{U^{(2,1)}} & \multigate{1}{U^{(3,1)}} & \qw & \cdots & & \qw & \meterB{|x_3\rangle}\\
\lstick{|0\rangle} & \qw & \multigate{1}{U^{(1,2)}} & \qw & \ghost{U^{(3,1)}} & \qw & \cdots & & \qw & \meterB{|x_4\rangle}\\
\lstick{|0\rangle} & \qw & \ghost{U^{(1,2)}} & \qw & \qw & \qw & \cdots & & \qw & \meterB{|x_5\rangle}
}}\par
    \caption{An example 2-local circuit. $U^{(i,j)}$ denotes the $j^{\textrm{th}}$ $2$-qudit unitary in the $i^{\textrm{th}}$ layer of the circuit}
    \label{fig:circuit}
\end{figure}
%
%

The main insight of this subsection is the following:

\begin{theorem}\label{ThmPseudoDimFixedQuantumCircuit}
With the notation and assumptions from above, it holds that
$$\Pseudo(\mathcal{F}_\mathcal{N})\leq 8d^4\cdot \gamma\cdot \log(16e\cdot\gamma).$$
\end{theorem}\par
Here and throughout the paper, $\log$ denotes the logarithm to base $2$.\par
To prove this result, we provide the following. 
\begin{lemma}\label{LmmPolynomialDescribingFixedQuantumCircuit}
With the notation and assumptions from above, there exists a polynomial $p_\mathcal{N}$ with real coefficients, in $2\gamma d^4 + 2dn$ real variables of degree $\leq 2(\gamma+n)$ such that every $f\in\mathcal{F}_\mathcal{N}$ can be obtained from $p_\mathcal{N}$ by fixing values for the first $2\gamma d^4$ variables. Moreover, in each term of $p$, the degree in the first $2\gamma d^4$ real variables is $\leq 2\gamma$ and the degree in the last $2dn$ real variables is $\leq 2n$.
\end{lemma}

\noindent Notably, there is no explicit dependence on depth $\delta.$

\noindent \begin{proof}
We first observe that
\begin{align*}
|\langle x|U_\mathcal{N}|0\rangle^{\otimes n}|^2 = |\langle 0|^{\otimes n} U_\mathcal{N}^\dagger |x\rangle|^2.
\end{align*}
We study this expression in a layer-wise analysis. When reading the circuit from right to left, 
the state that enters layer $\delta$ is transformed by the unitary $\bigotimes\limits_{j=1}^{\gamma_\delta} U^{(\delta,j)\dagger}$ such that each amplitude of the state after the $\delta^{th}$ layer is a linear combination of the amplitudes of $|x\rangle$, where each coefficient is a multilinear monomial of degree $\gamma_\delta$ in some of the $\gamma_\delta \cdot d^4$ complex entries of the $\lbrace U^{(\delta,j)\dagger}\rbrace_{1\leq j\leq \gamma_\delta}$.\par
%
By iterating this reasoning, we see that the state after the $(\delta - i)^{th}$ layer has amplitudes which are given by a linear combination of the amplitudes of $|x\rangle$, where each coefficient is a multilinear polynomial of degree $\leq \sum\limits_{k=0}^i \gamma_{\delta - k}$ in (some of) the entries of the unitaries $\lbrace U^{(\delta - k,j_k)\dagger}\rbrace_{0\leq k\leq i, 1\leq j_k\leq \gamma_k}$.\par
In particular, the $|0\rangle^{\otimes n}$-amplitude of the state $U_\mathcal{N}^\dagger |x\rangle$ can be written as a linear combination of the amplitudes of $|x\rangle$, where each coefficient is given by a multilinear polynomial $q_\mathcal{N}$ of degree $\leq \sum\limits_{k=0}^\delta \gamma_{\delta - k}=\gamma$ in (some of) the $\gamma\cdot d^4$ complex entries of the unitaries $\lbrace U^{(i,j_i)\dagger}\rbrace_{0\leq i\leq\delta, 1\leq j_i\leq \gamma_i}$.\par 
Recalling that the probability of observing outcome $|0\rangle^{\otimes n}$ is the square of the absolute value of the corresponding amplitude of $|x\rangle,$ we obtain from the polynomial $q_\mathcal{N}$ a polynomial $p_\mathcal{N} = |q_\mathcal{N}|^2$ that describes the output probabilities. As $q_\mathcal{N}$ has degree at most $\gamma$ in the $\gamma \cdot d^4$ complex parameters of the unitaries, $p_\mathcal{N}$ has degree at most $2\gamma$ in the corresponding $2\gamma \cdot d^4$ real parameters. Fixing these $2\gamma d^4$ parameters corresponds to fixing the circuit, and therefore one may obtain every $f \in \mathcal{F}_\mathcal{N}$ by fixing these parameters in $p_\mathcal{N}.$\par
Moreover, $p_\mathcal{N}$ is a polynomial in the $2dn$ real parameters which give rise to the amplitudes of $|x\rangle$. (Here, the assumption that $|x\rangle$ is a product state enters.) As each such amplitude has degree $\leq n$ in the $2dn$ complex parameters, the degree of $p_\mathcal{N}$ in these real parameters is at most $2n.$
\qed \end{proof}
\begin{remark}
We formulate the result only for measurement operators consisting of tensor products of $1$-dimensional projections, and continue to do so throughout this manuscript. For $x \in X,$ we can write $|x\rangle = \bigotimes\limits_{i=1}^n \left( \sum\limits_{j=0}^{d-1} \alpha^{(i)}_j |j\rangle\right)$, so we associate $dn$ complex variables with $x$. That each amplitude of $|x\rangle$ can be written as a product of $n$ complex parameters gives rise to the upper bound of $n$ in the degree.\par
We could instead look at more general measurement operators consisting of $1$-dimensional projections without requiring product structure, i.e.~entangled measurements. In this scenario, we would write $|x\rangle=\sum\limits_{z\in\{ 0,\ldots,d-1 \}^n } x_z |z\rangle$, associating $d^n$ complex variables with $x$. In this setup, each amplitude of $x$ is simply a polynomial of degree $1$ in these complex variables.\par
As we fix the variables corresponding to $x$ and $y$ in the shattering assumption that appears in our proof of Theorem \ref{ThmPseudoDimFixedQuantumCircuit}, their corresponding degrees are not relevant to our argument; only the degree in the entries of the unitaries enters our analysis. Therefore, both product measurements or entangled measurements lead to the same pseudo-dimension bound. This is due to the fact that allowing for entangled measurements changes the set of allowed inputs but not the function class itself.\end{remark}
Now that we have established Lemma \ref{LmmPolynomialDescribingFixedQuantumCircuit}, we can prove Theorem \ref{ThmPseudoDimFixedQuantumCircuit} with reasoning analogous to that in \citep{Goldberg.1995}.\\

\begin{proof} (Theorem \ref{ThmPseudoDimFixedQuantumCircuit})
Let $\lbrace (x_i,y_i)\rbrace_{i=1}^m\subseteq X\times\IR$ be such that for every $C\subseteq\lbrace 1,\ldots,m\rbrace$ there exists $f_C\in\mathcal{F}_\mathcal{N}$ such that 
$
f_C(x_i) - y_i\geq 0$ if and only if $i\in C.
$\par
By Lemma \ref{LmmPolynomialDescribingFixedQuantumCircuit}, there exists a polynomial $p_\mathcal{N}$ in $2\gamma d^4 + 2dn$ real variables of degree $\leq 2(\gamma + n)$ such that for every $C\subseteq\lbrace 1,\ldots,m\rbrace$ there exists an assignment $\Xi_C$ to the first $2\gamma d^4$ variables of $p_\mathcal{N}$ such that 
$
p_\mathcal{N}(\Xi_C,x_i) - y_i \geq 0$ if and only if $i\in C.
$\par
In particular, this implies (using the ``moreover'' part of Lemma \ref{LmmPolynomialDescribingFixedQuantumCircuit}) that the set 
$
\mathcal{P} = \lbrace p_\mathcal{N}(\cdot,x_i) - y_i\rbrace_{i=1}^m
$
is a set of $m$ polynomials of degree $\leq 2\gamma$ in $2\gamma d^4$ real variables that has at least $2^m$ different consistent sign assignments.\par
We now claim that $m\leq 8d^4\cdot\gamma\cdot\log(16e\cdot\gamma) $. If $m<2\gamma d^4$, this holds trivially. Hence, w.l.o.g.~$m\geq 2\gamma d^4$. So by Corollary \ref{CrlWarren}, we have
\begin{align*}
2^m \leq \left(\frac{8e\cdot 2\gamma\cdot m}{2\gamma d^4}\right)^{2\gamma d^4}.
\end{align*}\par
Taking logarithms now gives
\begin{align*}
m\leq 2\gamma d^4\left(\log(16e\cdot\gamma) + \log\left(\frac{m}{2\gamma d^4}\right)\right).
\end{align*}\\
Now we distinguish cases. If $16e\cdot\gamma\geq \frac{m}{2\gamma d^4}$, then the above immediately implies $m\leq 4\gamma d^4\cdot \log(16e\cdot\gamma)$. If $16e\cdot\gamma\leq \frac{m}{2\gamma d^4}$, then we obtain $m\leq 4\gamma d^4\cdot\log\left(\frac{m}{2\gamma d^4}\right)$, which in turn implies $m\leq 8\gamma d^4$. In both cases we have $m\leq 8d^4\cdot\gamma\cdot\log(16e\gamma)$. By definition of the pseudo-dimension, we conclude $\Pseudo(\mathcal{F}_\mathcal{N})\leq 8d^4\cdot \gamma\cdot \log(16e\gamma)$, as claimed.
\qed \end{proof}

The attentive reader may notice that we do not explicitly refer to the unitarity assumption in our reasoning; our argument mainly uses linearity. This already hints at a generalization to quantum circuits not of unitaries but of operations, which we will describe in Subsec.~\ref{SbSctQuOperations}. In that subsection, we will also see how the unitarity assumption implicit in this proof produces a better upper bound than in the general setting of quantum operations. 
\begin{remark}
We formulate our bounds in terms of the pseudo-dimension, not its scale-sensitive version called fat-shattering dimension, even though the latter is more commonly used in classical learning. In our scenario, however, the pseudo-dimension and the fat-shattering dimension effectively coincide. This is because we could apply our reasoning for general matrices instead of only unitaries in the setting of Theorem \ref{ThmPseudoDimFixedQuantumCircuit} as well and achieve the same bounds. In that case, however, the resulting real-valued function class is closed under scalar multiplication with non-negative scalars and it follows from the definition that for such classes, the fat-shattering dimension equals the pseudo-dimension.
\end{remark}
\subsection{Variable Circuit Structure}\label{pseudo-variable}
Whereas in the previous subsection we fixed a quantum circuit architecture and only varied the entries of the two-qudit unitaries plugged into this structure, we now additionally vary the structure of the quantum circuit architecture itself and consider the complexity of the class of all quantum circuits of a given depth and size. Once again, we consider 2-local quantum circuits, i.e.~circuits with one- and two-qudit gates acting on arbitrary pairs of qudits.\\

The class of states which is of relevance in this analysis is
\begin{align*}
\mathcal{S}_{\delta,\gamma}\left((\IC^d)^{\otimes n}\right)
:= \lbrace |\psi\rangle\ |\ \exists \textrm{ quantum circuit } \mathcal{N} \textrm{ of depth } \delta\textrm{ and size } \gamma \textrm{ such that } |\psi\rangle\in \mathcal{S}_\mathcal{N}\left( (\IC^d)^{\otimes n}\right)\rbrace.
\end{align*}
Again, this set of states gives rise to a function class via
\begin{align*}
\mathcal{F}_{\delta,\gamma}
:= \lbrace f:X\to [0,1]\ |\ \exists |\psi\rangle\in \mathcal{S}_{\delta,\gamma}\left((\IC^d)^{\otimes n}\right):~f(x)=|\langle x|\psi\rangle|^2\rbrace,
\end{align*}
where $X$ is as above given by $X=S_d\times\ldots\times S_d.$ As before, we want to bound the pseudo-dimension of this function class.\\

We summarize the result of this subsection in the following: 
\begin{theorem}\label{ThmPseudoDimVariableQuantumCircuit}
With the notation and assumptions from above, it holds that
$$\Pseudo(\mathcal{F}_{\delta,\gamma})\leq \mathcal{O} (\delta \cdot d^4 \cdot \gamma^2 \log \gamma).$$
\end{theorem}

As with Theorem \ref{ThmPseudoDimFixedQuantumCircuit}, the main step towards this result consists of relating the functions appearing in $\mathcal{F}_{\delta,\gamma}$ to polynomials. The difference here is that we must upper bound the number of polynomials, as below.

\begin{lemma}\label{LmmPolynomialsDescribingVariableQuantumCircuit}
With the notation and assumptions from above, there exists a set $\mathcal{P}_{\delta,\gamma}$ of polynomials with real coefficients, in $2\gamma d^4 + 2dn$ real variables of degree $\leq 2(\gamma+n)$ such that for every $f\in\mathcal{F}_{\delta,\gamma}$ there exists a polynomial $p\in\mathcal{P}_{\delta,\gamma}$ such that $f$ can be obtained from $p$ by fixing values for the first $2\gamma d^4$ variables, and such that
\begin{align*}
|\mathcal{P}_{\delta,\gamma}|
\leq\frac{\gamma! ~\delta^{\gamma-\delta}}{(\gamma-\delta)!}(n!)^\delta.
\end{align*}
Moreover, in each term of $p\in\mathcal{P}_{\delta,\gamma}$ the degree in the first $2\gamma d^4$ real variables is $\leq 2\gamma$ and the degree in the last $2dn$ real variables is $\leq 2n$.
\end{lemma}
\noindent \begin{proof}
There are at most $\frac{\gamma!~\delta^{\gamma-\delta}}{(\gamma-\delta)!}$ ways to assign the gates among the $\delta$ layers. The term $\frac{\gamma!}{(\gamma-\delta)!}$ counts assigning a single gate to each layer, to ensure that there are no trivial (empty) layers. Having assigned each layer one gate, the remaining $\gamma-\delta$ gates may be distributed to any of the $\delta$ layers.\\

Next, we bound the number of ways of assigning qudits to the circuit layers, so that the qudits are inputs to the fixed-position unitaries. For our purposes, it suffices to crudely upper bound this by $n!$ for each single layer and thus by $(n!)^\delta$ overall. Hence, there are at most
\begin{align*}
    \frac{\gamma! ~\delta^{\gamma-\delta}}{(\gamma-\delta)!}(n!)^\delta
\end{align*}
different quantum circuit architectures.
The proof is completed by applying Lemma \ref{LmmPolynomialDescribingFixedQuantumCircuit} to every such quantum circuit architecture.
\qed \end{proof}\par
Now that we have established Lemma \ref{LmmPolynomialsDescribingVariableQuantumCircuit}, we can prove Theorem \ref{ThmPseudoDimVariableQuantumCircuit} by reasoning analogous to that in \citep{Goldberg.1995}. See the Appendix for the proof of Theorem \ref{ThmPseudoDimVariableQuantumCircuit}.\\

\subsection{Extension to Circuits with Variable Inputs}\label{SbSctVariableInputs}
We now modify the results of \ref{pseudo-fixed} and \ref{pseudo-variable} to allow not only for the fixed input $|0\rangle^{\otimes n}$, but also for variable input. This is of use, for instance, in Subsec.~\ref{learnability}, in which we consider the PAC-learnability of quantum circuits (of unitary gates or more general quantum channels). In that context, allowing variable input amounts to learning the entire quantum circuit, rather than just its action on $|0\rangle^{\otimes n}.$ This is necessary in order to meaningfully compare the learning problem in Subsec.~\ref{learnability} to exact circuit tomography.\par
To consider variable input states, we define the following function classes, analogously to those in Subsecs.~\ref{pseudo-fixed} and \ref{pseudo-variable}: 
\begin{align*}
\mathcal{F}'_\mathcal{N}:= \lbrace & f:X\times Y\to [0,1]\ |\ \exists U_{\mathcal{N}|\lbrace U^{(i,j)}\rbrace},~U^{(i,j)}\in\mathcal{U}\left( (\IC^d)^{\otimes 2}\right): f(x,y)=|\langle x|U_{\mathcal{N}}|y\rangle|^2\rbrace,
\end{align*}

\noindent where $Y$ can be taken as the computational basis states $\{0,1,...,d-1\}^n,$ or more generally as $Y=X=S_d \times ... \times S_d.$

\begin{lemma}\label{LmmFixedQuantumCircuitVariableInput}
With the notation and assumptions from above the following holds: There exists a polynomial $p'_\mathcal{N}$ in $2\gamma d^4 + 4dn$ real variables of degree $\leq 2\gamma+4n$ such that  every $f\in\mathcal{F}'_\mathcal{N}$ can be obtained from $p'_\mathcal{N}$ by fixing values for the first $2\gamma d^4$ variables. Moreover, in each term of $p'_\mathcal{N}$ the degree in the first $2\gamma d^4$ real variables is $\leq 2\gamma$, the degree in the $2dn$ real variables corresponding to $x\in X$ is $\leq 2n$, and the degree in the $2dn$ real variables corresponding to $y\in Y$ is $\leq 2n$.
\end{lemma}

\noindent \begin{proof} Consider the product state input $|y\rangle = \sum_{z} y_z|z\rangle.$ As we consider product states, each $y_z$ is a product of $n$ complex parameters. Following the same reasoning as before, for a fixed $z\in\{ 0,\ldots,d-1\},$ $\langle z |U_\mathcal{N}|x\rangle$ is a multilinear polynomial $q_\mathcal{N}^{z}.$ Then, the amplitude $\langle y|U_\mathcal{N}|x\rangle$ is
\begin{align*}
q'_\mathcal{N}(x,y) = \langle y|U_\mathcal{N}|x\rangle &=\sum_{z\in\{0,1,...,d-1\}^n} \overline{y_z}~ 
    \langle z |U_{\mathcal{N}}|x\rangle\\
    &=\sum_{z\in\{0,1,...,d-1\}^n} \overline{y_z}~ q_\mathcal{N}^z (x).
\end{align*}

In the above equation, $q'_\mathcal{N}(x,y)$ has degree at most $n$ in $y,$ and so upon squaring the amplitude $q'_\mathcal{N}(x,y)$ to obtain $p'_\mathcal{N}(x,y)$ as in \ref{LmmPolynomialDescribingFixedQuantumCircuit}, we have a degree at most $2n$ in the $2dn$ real variables corresponding to $y.$ The rest follows from Lemma \ref{LmmPolynomialDescribingFixedQuantumCircuit}.\qed
\end{proof}

The bound from Theorem \ref{ThmPseudoDimFixedQuantumCircuit} still holds for the case of variable circuit input, with the proof proceeding almost identically upon replacing Lemma \ref{LmmPolynomialDescribingFixedQuantumCircuit} by Lemma \ref{LmmFixedQuantumCircuitVariableInput}. The $2d\cdot n$ additional variables that arise from the polynomial $y$- \mbox{dependence} do not alter the bound because we fix the values of these variables in the pseudo-shattering assumption.

\subsection{Extension to Circuits of Quantum Operations}\label{SbSctQuOperations}
We finish this section by describing an extension of Theorems \ref{ThmPseudoDimFixedQuantumCircuit} and \ref{ThmPseudoDimVariableQuantumCircuit} to the case of circuits of quantum operations, instead of only unitaries. This generalization is relatively straightforward because the decisive property of unitaries used in our previous proofs was not the preservation of inner products, but rather linearity. This setting is useful to e.g.~describe circuits with imperfect gates. Rather than consider a logical gate that implements a unitary exactly, each gate can instead be considered a quantum operation that executes the desired unitary with some probability, and e.g.~depolarizes input qudits with some probability. (Other noise models are of course possible.) Note that although quantum operations can, by Stinespring’s dilation theorem, be viewed as subsystem dynamics of a larger, unitarily-evolving system, if we only have access to measurement data for the subsystem then we cannot directly apply our result for the unitary case.\par
We use analogous notation to that introduced at the beginning of Subsec.~\ref{pseudo-fixed}, writing $T_{\mathcal{N}|\{T^{(i,j)}\}}$ for the overall quantum operation implemented by $\mathcal{N}$ when plugging the two-qudit quantum operations $\{T^{(i,j)}\}_{1\leq i\leq\delta, 1\leq j\leq\gamma_i}$ into the respective positions of the quantum circuit.\par
The quantum circuit $\mathcal{N}$ (of operations) now gives rise to the set of output states
\begin{align*}
\mathcal{D}_\mathcal{N}\left((\IC^d)^{\otimes n} \right) := \{&T_{\mathcal{N}|\{T^{(i,j)}\}}(|0^n\rangle\langle 0^n|)\ ~|~T^{(i,j)}\in\mathcal{T}\left((\IC^d)^{\otimes 2} \right) \},
\end{align*}
where we write $|0^n\rangle=|0\rangle^{\otimes n}$, so $|0^n\rangle\langle 0^n|=(|0\rangle\langle 0|)^{\otimes n}$.\\
By taking into account all possible quantum circuits of size $\gamma$ and depth $\delta$, we obtain
\begin{align*}
    \mathcal{D}_{\delta,\gamma}\left((\IC^d)^{\otimes n} \right)
    := \{\rho\ |\ &\exists~ \textrm{circuit}~\mathcal{N} \textrm{ of } \text{two-qudit operations}
    \text{of size }\gamma\\ &\textrm{ and depth }\delta\textrm{ such that }\rho\in\mathcal{D}_\mathcal{N}\left((\IC^d)^{\otimes n} \right) \}.
\end{align*}
These states now yield again a $p$-concept class
$$\mathcal{G}_{\delta,\gamma}:=  \{ f:X\to [0,1]\ |\ \exists\rho\in\mathcal{D}_{\delta,\gamma}\left((\IC^d)^{\otimes n} \right): f(x)=\langle x|\rho | x\rangle \}.$$
In this scenario, we show:
\begin{theorem}\label{ThmPseudoDimVariableQuantumCircuitOperations}
With the notation and assumptions from above, it holds that
$$
\Pseudo(\mathcal{G}_{\delta,\gamma})\leq \mathcal{O} (\delta \cdot d^8 \cdot \gamma^2 \log \gamma).
$$
\end{theorem}
\noindent \begin{proof}
We only sketch the reasoning, as it is similar to that in the proof of Theorem \ref{ThmPseudoDimVariableQuantumCircuit}. We first need to establish an analogue of Lemma \ref{LmmPolynomialsDescribingVariableQuantumCircuit}. To this end, observe that a quantum operation acting on two-qudit states can be interpreted as a $d^4\times d^4$ matrix with complex entries. Moreover, we may write
\begin{align*}
    \langle x| T_\mathcal{N}(|0^n\rangle\langle 0^n|) |x\rangle
    &= \tr[ T_\mathcal{N}(|0^n\rangle\langle 0^n|) |x\rangle\langle x|]\\
    &= \tr[ |0^n\rangle\langle 0^n| T^*_\mathcal{N}(|x\rangle\langle x|)]\\
    &=~ \langle 0^n| T^*_\mathcal{N}(|x\rangle\langle x|) |0^n\rangle,
\end{align*}
where $T^*_\mathcal{N}$ denotes the adjoint operation of $T_\mathcal{N}$ with regard to the Hilbert-Schmidt inner product.\par
As before, we can do a layer-wise analysis of the transformation of $|x\rangle\langle x|$ and observe that the entries of the (sub-normalized) density matrix after a layer can be written as linear combinations of the entries of the (sub-normalized) density matrix before the layer. Moreover, the coefficients can be written as multilinear polynomials with the degree determined by the number of two-qudit operations in the layer. Hence, we obtain the result of Lemma \ref{LmmPolynomialDescribingFixedQuantumCircuit} with $d^8$ instead of $d^4$. The bound on the number of different quantum circuit architectures can be derived in exactly the same way as before, so the analogue of Lemma \ref{LmmPolynomialsDescribingVariableQuantumCircuit} holds, completing the proof of the theorem.
\qed \end{proof}

Theorem \ref{ThmPseudoDimVariableQuantumCircuitOperations} and its proof sketch also help to elucidate the relevance of the unitarity assumption in Theorems \ref{ThmPseudoDimFixedQuantumCircuit} and \ref{ThmPseudoDimVariableQuantumCircuit}. 
Unitarity justifies our restriction to pure states, but in other respects Theorems \ref{ThmPseudoDimFixedQuantumCircuit} and \ref{ThmPseudoDimVariableQuantumCircuit} do not exploit unitary. The difference between Theorems \ref{ThmPseudoDimVariableQuantumCircuit} and \ref{ThmPseudoDimVariableQuantumCircuitOperations} amounts to the size of the matrices that represent the unitaries or quantum operations.
\section{Applications}\label{sec:applications}
In this section, we explore two different applications of our pseudo-dimension upper bounds. First, we employ the pseudo-dimension 
to exhibit a large but finite discrete set of quantum states, out of which at least one is hard to implement in the sense that preparing it requires exponentially many $2$-qubit unitaries. Second, we combine the pseudo-dimension bound with results from the theory of $p$-concept learning to derive the PAC-learnability of quantum circuits.\par
\subsection{Lower Bounds on the Gate Complexity of Quantum State Preparation}\label{sec:lower}
It is well known that almost all $n$-qubit unitaries require an exponential (in $n$) number of $2$-qubit unitaries to be implemented. Similarly, almost all pure $n$-qubit states require an application of exponentially (in $n$) many $2$-qubit unitaries to be generated from the $|0\rangle^{\otimes n}$ state \citep[see e.g.][]{Nielsen.2010}. However, in neither case are there explicit examples of unitaries or states saturating this exponentiality bound. (See \citep{Aaronson.2016} for more information on the gate complexity of unitary implementation and state preparation.) We will use the pseudo-dimension as a tool to exhibit a discrete set of pure qubit states such that at least one of them requires exponentially many $2$-qubit unitaries to be generated from $|0\rangle^{\otimes n}$.\par
The drawback of our result is that the size of this set is $2^{2^n}$ and thus unsatisfyingly large. By relatively simple deliberations this size can be reduced by an order of $2^n$ elements, though this is negligible compared to the overall size.\par
We now describe the construction of the candidate set of states. For a subset $C\subseteq\lbrace |x0\rangle\rbrace_{x\in\lbrace 0,1\rbrace^n}$, namely a subset of the set of all computational basis states of $n+1$ qubits that end on $0$, with $C\neq\emptyset$, define
\begin{align*}
|\psi_C\rangle = \frac{1}{\sqrt{|C|}}\sum\limits_{x0\in C} |x0\rangle.
\end{align*}
For $C=\emptyset$ we take
\begin{align*}
|\psi_\emptyset\rangle = |0\rangle^{\otimes n}\otimes |1\rangle.
\end{align*}
(Note that the $(n+1)^{st}$ qubit only really matters for $|\psi_\emptyset\rangle$.) Our set of interest will be $$\mathcal{S}:=\lbrace |\psi_C\rangle \ |\ C\subseteq\lbrace |x0\rangle\rbrace_{x\in\lbrace 0,1\rbrace^n}\rbrace.$$

This discrete set of $2^{2^n}$ multi-qubit quantum states now gives rise to a class of $p$-concepts
\begin{align*}
    \mathcal{F}_\mathcal{S} = \lbrace f_C:X\to [0,1]\ |\ \exists C\subseteq\lbrace |x0\rangle\rbrace_{x\in\lbrace 0,1\rbrace^n}:~ f_C(x) = |\langle x|\psi_C\rangle|^2\rbrace.
\end{align*}

\noindent This class has large pseudo-dimension, as described in the following lemma.

\begin{lemma}\label{LmmPseudoDimLowerBoundCandidateClassV2}
With the notation introduced above, it holds that
$
\Pseudo(\mathcal{F}_\mathcal{S}) \geq 2^n.
$
\end{lemma}
\noindent \begin{proof}
Consider the subset of computational basis states\\ $\lbrace |x0\rangle\rbrace_{x\in\lbrace 0,1\rbrace^n}$ and the corresponding threshold values $y_{x0}=\frac{1}{2^n}=\min\limits_{C\subseteq\lbrace |x0\rangle\rbrace_{x\in\lbrace 0,1\rbrace^n}}\frac{1}{\lvert C \rvert} $ independently of $x0$. By construction of $\mathcal{S}$ and thus $\mathcal{F}_\mathcal{S}$ the following holds: \\For any $C\subseteq\lbrace |x0\rangle\rbrace_{x\in\lbrace 0,1\rbrace^n}$
\begin{align*}
    f_C(x0) = |\langle x0|\psi_C\rangle|^2 = 
    \begin{cases}
    \frac{1}{\lvert C \rvert}\quad &\textrm{if }|x0\rangle\in C\\
    0 &\textrm{else}
    \end{cases}.
\end{align*}
In particular, we have
\begin{align*}
f_C(x0) \geq y_{x0}\ \Longleftrightarrow\ |x0\rangle\in C.
\end{align*}
Hence, $\Pseudo(\mathcal{F}_\mathcal{S})\geq 2^n$, because we have found an example of a set of size $2^n$ that is pseudo-shattered.
\qed \end{proof}

We now combine this simple observation with Theorem \ref{ThmPseudoDimVariableQuantumCircuit}, which gives us the following:

\begin{theorem}
With the notation introduced above, if $\gamma$ and $\delta$ are such that each state in $\mathcal{S}$ can be generated from the state $|0\rangle^{\otimes (n+1)}$ by some circuit of size $\gamma$ and depth $\delta$, then 
\begin{align*}
2^n \leq \mathcal{O}\left(\delta \cdot 2^4 \cdot \gamma^2 \log \gamma \right)
\end{align*}
\end{theorem}
\noindent \begin{proof}
Under the assumption of the Theorem we can conclude $\mathcal{F}_\mathcal{S}\subseteq\mathcal{F}_{\delta,\gamma}$. Now combine the lower bound of Lemma \ref{LmmPseudoDimLowerBoundCandidateClassV2} with the upper bound from Theorem \ref{ThmPseudoDimVariableQuantumCircuit}.
\qed \end{proof}

\begin{corollary}\label{CrlCandidateStatesForHighComplexity}
There exists a $C\subseteq \lbrace |x0\rangle\rbrace_{x\in\lbrace 0,1\rbrace^n}$ such that $|\psi_C\rangle= \frac{1}{\sqrt{|C|}}\sum\limits_{|x0\rangle\in C}|x0\rangle$ cannot be implemented by a quantum circuit of $2$-qubit unitaries with subexponential (in $n$) size or depth.
\end{corollary}

Note that any set of functions which pseudo-shatters a set of size $2^n$ has to have at least $2^{2^n}$ elements. Hence, the large size of the set $C$ is an automatic consequence of our line of reasoning.\par
\begin{remark}
We note that a set of $n$-qubit states with cardinality doubly exponential in $n$ s.t.~at least one of them needs an exponential number of gates (up to logarithmic factors) to be implemented can also be obtained with more standard reasoning. Namely, it is well known that there are $n$-qubit states the approximation of which up to trace-distance $\varepsilon$ requires $\Omega\left( \frac{2^n\log\left(\tfrac{1}{\varepsilon}\right)}{\log(n)}\right)$ unitary gates \citep[see][chap.~4.5.4]{Nielsen.2010}. So if we pick a $\frac{1}{2}$-net of size $\mathcal{O}\left( 2^{2^n} \right)$ for the set of pure $n$-qubit quantum states, this will have the desired properties.
\end{remark}
\indent\indent We sketch another way of using our pseudo-dimension bound to study the gate complexity of state preparation and which might lead to a smaller set of candidates. Given $n$-qubit pure states $|\psi_1\rangle,\ldots,|\psi_m\rangle$ and efficiently implementable (i.e. with polynomially many $2$-qubit unitary gates arranged in polynomially many layers) unitaries $U_1,\ldots,U_k$, one can study the set of states
$\{ U_i|\psi_j\rangle\}_{1\leq i\leq k, 1\leq j\leq m}$.\\ If an exponential (in $n$) pseudo-dimension lower bound can be established for $$\{f:X\to [0,1] ~|~ \exists 1\leq i\leq k, 1\leq j\leq m: f (x)=|\langle x|U_i|\psi_j\rangle|^2\},$$ then, since every $U_i$ is efficiently implementable, one can conclude that at least one among the states $|\psi_j\rangle$ is not efficiently implementable.\par
The advantage of such a pseudo-dimension-based reasoning would be that $m$ need not be doubly exponential in $n$, since we can compensate for this in $k$. This realization can already be used to reduce the size of the set of candidate states given in Corollary \ref{CrlCandidateStatesForHighComplexity}. However, we have not yet been able to identify sufficiently many efficiently-implementable unitaries to reduce the size below doubly exponential. Nevertheless, there is likely room for improvement in applying our method to the gate complexity of quantum state preparation.
\subsection{Learnability of Quantum Circuits}\label{learnability}
We now use our pseudo-dimension bounds to study learnability. Specifically, we use the pseudo-dimension bound for the case of variable inputs (Subsec.~\ref{SbSctVariableInputs}) combined with the generalization to quantum operations (Subsec.~\ref{SbSctQuOperations}). We proceed quite similarly to \citep{Aaronson.2007}.\\

The learning problem which we want to study is the following:
Let $\mu$ be a probability measure on $(X\times Y)\times [0,1]$, unknown to the learner. Let $S=\lbrace ((x^{(i)},y^{(i)}),p^{(i)})\rbrace_{i=1}^m$ be corresponding training data drawn i.i.d.~according to $\mu$. A learner must, upon input of training data $S$, size $\Gamma\in\IN$, depth $\Delta\in\IN$, confidence $\delta\in [0,1)$, accuracy, $\varepsilon\in [0,1)$ and error margin $\beta\in (0,1)$, output a hypothesis quantum circuit $\mathcal{N}$ of size $\Gamma$ and depth $\Delta$ consisting of two-qudit operations such that, with probability $\geq 1-\delta$ with regard to the choice of training data,  
\begin{align*}
    \mathbb{P}_{((x,y),p)\sim\mu}\left[ |f_\mathcal{N}(x,y) - p|>\beta \right]
    \leq\varepsilon + \inf\limits_{\mathcal{M}} \mathbb{P}_{(x,p)\sim\mu}\left[ |f_\mathcal{M}(x,y) - p|>\beta \right],
\end{align*}
where the infimum runs over all quantum circuits $\mathcal{M}$ of size $\Gamma$ and depth $\Delta$. Here, $f_\mathcal{N}$ denotes the function $f_\mathcal{N}(x,y)=\langle x| T_\mathcal{N}(|y\rangle\langle y|) |x\rangle$ and $f_\mathcal{M}$ is defined analogously, similarly to Subsec.~\ref{SbSctVariableInputs}.\par
We use our pseudo-dimension bound in order to upper-bound the size of the training data sufficient for solving this task. More precisely, we make use of sample complexity upper bounds from the fat-shattering dimension as proved in \citep{Anthony.2000,Bartlett.1998}, together with the fact that the fat-shattering dimension is upper-bounded by the pseudo-dimension. \par
First we restrict our scope to the ``realizable'' scenario, i.e.~we will assume the probability measure to be of the form 
\begin{align*}
    \mu((x,y),p) = \begin{cases} \mu_1(x,y)\quad &\textrm{if } p = f_{\mathcal{N}_*}(x,y) \\ 0 &\textrm{else} \end{cases}
\end{align*}
for some quantum circuit $\mathcal{N}_*$ of size $\Gamma$ and depth $\Delta$. This will in particular imply that for quantum circuits $\mathcal{M}$ of size $\Gamma$ and depth $\Delta$
\begin{align*}\inf\limits_{\mathcal{M}} \mathbb{P}_{((x,y),p)\sim\mu}\left[ |f_\mathcal{M}(x,y) - p|>\beta \right]=0.
\end{align*}
Colloquially, realizability means that there exists a set of ``correct'' parameters $\Gamma$ and $\Delta$ and these are known to the learner, i.e.~training samples are promised to be drawn from circuits of size $\Gamma$ and depth $\Delta$.\par
We will focus on a proper learning scenario, i.e., we will assume the unknown target circuit to be in some (known) class, namely the class of circuits whose size and depth satisfy certain polynomial bounds, and require the learner to output an element of that same class as hypothesis.\par
We will make use of the following classical result:
\begin{theorem}\label{ThmSampleComplexityViaFatShattering} \emph{\citep[][Corollary $3.3$]{Anthony.2000}}\\
Let $X$ be an input space, let $\mathcal{F}\subseteq [0,1]^X$. Let $D$ be a probability measure on $X$, let $f_*\in\mathcal{F}$. Let $\delta,\varepsilon,\alpha,\beta\in (0,1)$ with $\beta >\alpha$. Let $\mathcal{S}=\lbrace x_1,\ldots,x_m\rbrace$ be a set of $m$ samples drawn i.i.d.~according to $D$. Let $h\in\mathcal{F}$ be such that $|h(x_i)-f_*(x_i)|\leq\alpha$ for all $1\leq i\leq m$.\\
Then, a sample size\\
$$m=\mathcal{O}\left(\frac{1}{\varepsilon}\left( \fat_\mathcal{F}\left( \frac{\beta-\alpha}{8}\right)\log^2\left(\frac{\fat_\mathcal{F}\left(\frac{\beta-\alpha}{8}\right)}{(\beta - \alpha)\varepsilon}\right) + \log\frac{1}{\delta}\right)\right)$$ suffices to guarantee that, with probability $\geq 1-\delta$ with regard to the choice of training data $\mathcal{S}$,
\begin{align*}
    \mathbb{P}_{x\sim D}[ |h(x)-f_*(x)|>\beta]\leq \varepsilon.
\end{align*}
\end{theorem}

In our setting, this result implies:
\begin{corollary}
Let $\mathcal{N}_*$ be a quantum circuit of quantum operations with size $\Gamma$ and depth $\Delta$. Let $\mu$ be probability measure on $X\times Y$ unknown to the learner. Let $$S=\lbrace ((x^{(i)},y^{(i)}),f_{\mathcal{N}_*}(x^{(i)},y^{(i)})\rbrace_{i=1}^m$$ be corresponding training data drawn i.i.d.~according to $\mu$. Let $\delta, \varepsilon, \alpha, \beta\in (0,1)$.
Then, training data of size $m=\mathcal{O}\left(\frac{1}{\varepsilon}\left( \Delta d^8\Gamma^2\log(\Gamma)\log^2\left(\frac{\Delta d^8\Gamma^2\log(\Gamma)}{(\beta - \alpha)\varepsilon}\right) + \log\frac{1}{\delta}\right) \right)$ suffice to guarantee that, with probability $\geq 1-\delta$ with regard to choice of the training data, any quantum circuit $\mathcal{N}$ of size $\Gamma$ and depth $\Delta$ that satisfies 
\begin{align*}
    |f_\mathcal{N}(x_i,y_i) - f_{\mathcal{N}_*}(x_i,y_i)|\leq\alpha\quad\forall 1\leq i\leq m
\end{align*}
also satisfies
\begin{align*}
    \mathbb{P}_{(x,y)\sim\mu}[|f_\mathcal{N}(x,y) - f_{\mathcal{N}_*}(x,y)|>\beta]\leq\varepsilon.
\end{align*}
\end{corollary}

\begin{proof}
Combine Theorem \ref{ThmSampleComplexityViaFatShattering} with Theorem \ref{ThmPseudoDimVariableQuantumCircuit} (more precisely, with its version for variable input states, which can be proved for operations analogously to the reasoning in Subsec.~\ref{SbSctVariableInputs}) and use that the fat-shattering dimension is always upper-bounded by the pseudo-dimension.
\end{proof}

Note that in particular, this implies that for the class of circuits of quantum operations with polynomial size and depth in the number of qudits, a hypothesis that performs well on training data will also perform well in a probably approximately correct sense.

Next, we want to discuss briefly how our result compares to the work \citep{Aaronson.2007} on the learnability of quantum states. There, it is shown that quantum states can be PAC-learned with a sample complexity that depends linearly on the number of qubits and (among other dependencies) polynomially on $\frac{1}{\varepsilon}$, where $\varepsilon$ denotes the desired accuracy. However, this result does not imply learnability of quantum channels with a sample complexity that depends polynomially on the number of qubits. This observation is already stated in \citep{Aaronson.2007}, and we provide an alternate, intuitive explanation for why the result on states does not directly apply to operations.

One can straightforwardly apply the result of \citep{Aaronson.2007} to learn the Choi-Jamiolkowski state of a quantum channel. One can then compute measurement probabilities of output states of a channel $T$ acting on $n$-qubit states, using its Choi-Jamiolkowski state $\tau$. For this we must make use of the formula $$\tr[ E T(\rho)] = 2^n \tr[\tau (E\otimes \rho^T)].$$ Here, we see that any error on the side of the Choi-Jamiolkowski state will be multiplied by a factor exponential in $n$, and thus in this case the overall $n$-dependence of the sample complexity bound from \citep{Aaronson.2007} becomes exponential via the accuracy-dependence.

This motivates our study of learnability of a restricted class of quantum operations. Finding such operations for which process tomography is possible was left as an open problem in \citep{Aaronson.2007}. Our answer to this question is that a PAC-version of quantum process tomography is possible when we restrict our scope to operations that can be implemented by quantum circuits of depth and size polynomial in the number of qudits. However, note that this is subject to a realizability assumption: the learner must known in advance a polynomial bound on the size and depth of the circuit. We show that imposing the operations be efficiently implementable automatically reduces the information-theoretic complexity of learning, requiring only a modest number of training examples. We do not make any statement about the computational complexity of this learning task; this remains an open problem. 

How can this probably approximately correct version of quantum process tomography be put to use? Given polynomially many uses of a black box implementing an unknown quantum operation of polynomial size and depth, one can exhibit a circuit of two-qudit quantum operations that approximates the unknown channel. In other words, we obtain a classical description of an approximate copy of the channel.
\section{Open Problems}\label{sec:discussion}
Finally, in this section we discuss future directions and possible generalizations of our results.\par
Two natural parameters of a circuit, depth and size, appear polynomially in the pseudo-dimension upper bounds. Notably, these bounds are independent of the number of qudits in the circuit. Are our upper bounds tight in their dependence on size and depth? Can similar techniques produce pseudo-dimension lower bounds? For example, by considering a single $2$-qudit unitary it is relatively straightforward to see that the pseudo-dimension of a circuit is $\geq\Omega (d)$. Can we close the gap in dimension-dependence between this linear lower bound and our quartic upper bound?\par
Our application of pseudo-dimension for lower bounds on the gate complexity of state preparation complements known methods \citep[described e.g.~in][]{Nielsen.2010}, based on counting dimensions or covering arguments. We exhibit a class of states of size $2^{2^n},$ for which at least one has exponential gate complexity of state preparation. Can we exploit this new technique to exhibit a smaller set of states? Perhaps the most exciting application of pseudo-dimension bounds could be provable lower bounds on the gate complexity of state preparation, if the reasoning in $\ref{sec:lower}$ is sharpened or the tools are developed further.\par
If circuit depth and size are known in advance, one can information-efficiently learn the circuit. If the learner receives training data generated by an approximation of the circuit, does the result still hold? Can the realizability assumption be relaxed?\par
Does ``pretty-good circuit tomography’’ have applications? On the theory side, this might involve exploiting the learning process as an approximate copy-machine for quantum circuits. 
Of interest for both theory and experiment is whether circuits can be learned with a reasonable amount of computation. One can imagine progress on this question for process tomography similar to that for state tomography; demonstrating a class of states for which learning is computationally efficient in \citep{Rocchetto.2017} made it possible to learn physically interesting states in a laboratory in \citep{Rocchetto.2019}. An efficiency improvement in the process tomography case might also have experimental ramifications.


\section{Acknowledgements}
M.C.C.~and I.D.~thank Michael Wolf for suggesting this problem and both Michael Wolf and Yifan Jia for insightful discussions. Also, M.C.C.~and I.D.~thank Scott Aaronson, Srinivasan Arunachalam and Andrea Rocchetto for their valuable feedback on an earlier version of this paper. Finally, M.C.C.~and I.D.~thank the reviewers for their helpful suggestions.

M.C.C.~gratefully acknowledges support from the TopMath Graduate Center of the TUM Graduate School at the Technische Universit{\"a}t M{\"u}nchen, Germany, and from the TopMath Program at the Elite Network of Bavaria. M.C.C.~is supported by a doctoral scholarship of the German Academic Scholarship Foundation (Studienstiftung des deutschen Volkes).

I.D.~gratefully acknowledges that this material is based upon work supported by the National Science Foundation (NSF) Graduate Research Fellowship under Grant No.~DGE 1656518, and by the German Academic Exchange Service (DAAD) under Grant No.~57381410. Any conclusions expressed in this material are those of the authors and do not necessarily reflect the views of the aforementioned institutions.

\newpage
\bibliographystyle{spbasic}
\bibliography{Literature.bib}

\newpage
\appendix
\section{Proof of Theorem \ref{ThmPseudoDimVariableQuantumCircuit}}
Here, we prove Theorem \ref{ThmPseudoDimVariableQuantumCircuit}, namely that 
$\Pseudo(\mathcal{F}_{\delta,\gamma})\leq \mathcal{O} (\delta \cdot d^4 \cdot \gamma^2 \log \gamma).$\\

\noindent \begin{proof}~(Theorem \ref{ThmPseudoDimVariableQuantumCircuit})\\
We rely upon Lemma \ref{LmmPolynomialsDescribingVariableQuantumCircuit}. 
Let $\lbrace (x_i,y_i)\rbrace_{i=1}^m\subseteq X\times\IR$ be such that for every $C\subseteq\lbrace 1,\ldots,m\rbrace$, there exists $f_C\in\mathcal{F}_{\delta,\gamma}$ such that 
$f_C(x_i) - y_i\geq 0$ if and only if $i\in C.$

By Lemma \ref{LmmPolynomialsDescribingVariableQuantumCircuit}, there exists a set of polynomials $\mathcal{P}_{\delta,\gamma}$ in $2\gamma d^4 + 2d^n$ real variables such that $|\mathcal{P}_{\delta,\gamma}|\leq \frac{\gamma! ~ \delta^{\gamma-\delta}}{(\gamma-\delta)!}(n!)^\delta$ and such that for every $C\subseteq\lbrace 1,\ldots,m\rbrace$, there exists a $p_C\in\mathcal{P}_{\delta,\gamma}$ and an assignment $\Xi_C$ to the first $2\gamma d^4$ variables of $p_C$ such that 
$p_C(\Xi_C,x_i) - y_i \geq 0$ if and only if $i\in C.$\\

In particular, this implies (using the ``moreover''-part of Lemma \ref{LmmPolynomialsDescribingVariableQuantumCircuit}) that the set 
$
\mathcal{P} = \lbrace p(\cdot,x_i) - y_i\rbrace_{i=1}^m\ |\ p\in\mathcal{P}_{\delta,\gamma}\rbrace
$
is a set of $m\cdot |\mathcal{P}_{\delta,\gamma}|\leq m\frac{\gamma! ~ \delta^{\gamma-\delta}}{(\gamma-\delta)!} ~ (n!)^\delta$ polynomials of degree $\leq 2\gamma$ in $2\gamma d^4$ real variables that has at least $2^m$ different consistent sign assignments. So by Corollary \ref{CrlWarren}, we have \begin{align*}
2^m \leq \left(\frac{8e\cdot2\gamma\cdot m}{2\gamma d^4} \cdot \frac{\gamma! ~\delta^{\gamma-\delta}}{(\gamma-\delta)!}~ (n!)^\delta \right)^{2\gamma d^4}.
\end{align*}

Taking logarithms yields
\begin{align*}
m\leq 2\gamma d^4\left(\log(16e\cdot\gamma) + \log\left(\frac{m}{2\gamma d^4} \cdot \frac{\gamma! ~ \delta^{\gamma-\delta}}{(\gamma-\delta)!}~ (n!)^\delta\right)\right).
\end{align*}
Repeating the argument in the proof of Theorem \ref{ThmPseudoDimFixedQuantumCircuit}, we distinguish cases and observe that in both cases, $$m\leq 8d^4\cdot\gamma\cdot\log\left(16e\gamma\cdot \frac{\gamma! ~ \delta^{\gamma-\delta}}{(\gamma-\delta)!}~ (n!)^\delta\right).$$

Expanding the logarithm and using Stirling's formula up to two terms, we have
\begin{align*}
&\log\left(16e\gamma\cdot \frac{\gamma! ~ \delta^{\gamma-\delta}}{(\gamma-\delta)!}~ (n!)^\delta \right)\\
&=\frac{1}{\ln 2}\Big(4\ln2+1+\ln\gamma+ \Big[n\cdot\delta \ln n - n\cdot \delta + \mathcal{O}(\ln n) +
\gamma\ln\gamma -\cancel{\gamma} + \mathcal{O}(\ln \gamma) - (\gamma-\delta)\ln(\gamma-\delta)\\
&\hphantom{\frac{1}{\ln 2}\Big(4\ln2+1+\ln\gamma+\Big[~}\ +(\cancel{\gamma}-\delta) + \mathcal{O}(\ln (\gamma-\delta))+ (\gamma-\delta)\ln\delta\Big]\Big)\\
&\leq \frac{1}{\ln 2}\Big(4\ln2+1 +\ln\gamma+ \left[ 2\gamma\cdot \delta(\ln(2\gamma)-1) + 
\gamma\ln\gamma -(\gamma-\delta)\ln(\gamma-\delta) -\delta + (\gamma-\delta)\ln\delta\right]\Big)\\
&=\mathcal{O}(\gamma\cdot\delta\log\gamma).
\end{align*}
We use the fact that $n\leq 2\gamma$ (because we assume that each qudit is acted upon by at least one gate) in the second step, and note that because $\gamma\geq \delta,$ the asymptotic behavior of all of the above terms are subsumed by the first term in the bracket. We have also confirmed that the $\log(16e\gamma)$ term above may be neglected. Thus, by the definition of the pseudo-dimension we conclude $\Pseudo(\mathcal{F}_\mathcal{N}) 
\leq \mathcal{O} (\delta \cdot d^4 \cdot \gamma^2 \log \gamma)$.
\qed \end{proof}
\end{document}